\documentclass{article}

\usepackage[french,english]{babel}
\usepackage [vscale=0.76,includehead]{geometry}
\usepackage{relsize}
\usepackage{graphicx,amsmath,fullpage,mathptmx,helvet, float, subcaption}
\usepackage[utf8]{inputenc}
\usepackage[T1]{fontenc}
\usepackage{amsmath, amssymb, amsthm, mathrsfs, amsfonts}
\usepackage{comment}
\usepackage{enumerate, enumitem}

\usepackage{graphicx} 
\usepackage{algorithm}
\usepackage{algpseudocode} 

\usepackage{environ,tabularx}

\usepackage{hyperref}
\hypersetup{
  pdftitle = {},
  pdfauthor = {TODO},
  colorlinks = true,
  linkcolor = black!30!red,
  citecolor = black!30!green
}

\usepackage{mathtools}
\newcommand{\defeq}{\vcentcolon=}

\newcommand{\set}[2]{   \left\{ #1 \ \big| \ #2 \right\}   }

\newtheorem{prop}{Proposition}
\newtheorem{lem}[prop]{Lemma}

\newtheorem{thm}[prop]{Theorem} 
\newtheorem{cor}[prop]{Corollary}

\theoremstyle{definition}  
\newtheorem*{rem}{Remark} 

\makeatletter
\newcommand{\problemtitle}[1]{\gdef\@problemtitle{#1}}
\newcommand{\probleminput}[1]{\gdef\@probleminput{#1}}
\newcommand{\problemquestion}[1]{\gdef\@problemquestion{#1}}
\NewEnviron{problem}{
  \problemtitle{}\probleminput{}\problemquestion{}
  \BODY
  \par\addvspace{.5\baselineskip}
  \noindent
  \begin{tabularx}{\textwidth}{@{\hspace{\parindent}} l X c}
    \multicolumn{2}{@{\hspace{\parindent}}l}{\sc \@problemtitle} \\
    \textbf{Input:} & \@probleminput \\
    \textbf{Question:} & \@problemquestion
  \end{tabularx}
  \par\addvspace{.5\baselineskip}
}
\makeatother
\usepackage{marginnote}
\usepackage[textsize=small, color=red!30]{todonotes}

\title{Halfspace separation in geodesic convexity}
\author{Niranjan Nair}
\date{October 2025}

\begin{document}

\maketitle
\begin{abstract}
Let $G = V, E$ be a simple connected undirected graph.  A set $X \subseteq V$ is \emph{geodesically convex} if for any pair of vertices $x, y \in X$, all vertices on all shortest paths in $G$ from $x$ to $y$ are contained in $X$. A set $H \subseteq V$ is said to be a {halfspace} if both $H$ and its complement (denoted by $H^c$) are convex. Given two sets $A, B \subseteq V$, the { halfspace separation} problem asks if there exist complementary halfspaces $H, H^c$ such that $A \subseteq H$ and $B \subseteq H^c$. The halfspace separation problem is known to be NP-complete for the geodesic convexity of general graphs. We show that geodesic halfspace separation is polynomial for weakly bridged graphs, pseudo-modular graphs, and the basis graphs of matroids.
\end{abstract}
\section{Introduction}
 A convexity space is a pair $(X, \mathfrak{C})$  where $X$ is a set and $\mathfrak{C}$ is a collection of subsets of $X$ such that $\emptyset, X \in \mathfrak{C}$ and $\mathfrak{C}$ is closed under taking intersection. Graphs provide several notions of convexity, such as geodesic (shortest path) and monophonic (induced path) convexity. A halfspace in a convexity space is a convex set whose complement is also convex. Halfspaces in linear
spaces arise in separation theorems, which are fundamental mathematical results with numerous
applications, in particular, in optimization and machine learning. The halfspace separation problem, introduced in the paper \cite{SHW23}, asks whether two sets $A$ and $B$ in a convexity space $(X, \mathfrak{C})$ can be separated by complementary halfspaces. Halfspace separation is known to be NP-complete in the geodesic convexity of graphs \cite{SHW23}. However, it can be solved in polynomial time for monophonic convexity. This was shown independently in \cite{Che24, BET24} and \cite{ENV24}. The approaches in \cite{Che24} and \cite{BET24} were combined in \cite{BCET25}, and reduces the problem to the separation of some special sets using the \emph{three step method}, which can then be formulated using {\sc 2-SAT}. The paper \cite{Che24} also gives an algorithm for geodesic separation by the method of halfspace enumeration. They show that the VC-dimension of the set of geodesic halfspaces of a graph $G$ is at most the clique number $\omega(G)$ when $G$ is chordal or Helly. The \emph{Sauer lemma} [insert ref] asserts that a set family of VC-dimension $d$ contains at most $O(n^d)$ sets. This implies that chordal and Helly graphs have at most $O(n^\omega)$ halfspaces. Thus we can enumerate all halfspaces in $O(\text{poly}(n)n^{\omega (G)})$ if $G$ is chordal or Helly. The paper \cite{Che24} also asks whether the complexity of halfspace enumeration can be improved to output polynomial time for chordal, Helly, and matroid basis graphs. 

We adapt the three step method introduced in \cite{Che24} to prove the following theorem.
\begin{thm}
\label{thm:main}
    Geodesic halfspace separation can be solved in polynomial time for weakly bridged graphs, pseudo-modular graphs, and basis graphs of matroids.
\end{thm}

These classes of graphs are well studied in the field of metric graph theory \cite{MGT}.

Bridged graphs, which were introduced in \cite{SoCh} and \cite{FJ87}, are the graphs in which all isometric cycles have length three. In particular, chordal graphs are bridged. Bridged graphs are also exactly the graphs in which all balls centred around convex sets are convex. \cite{Che00} presented a local-to-global characterization of bridged graphs as the 1-skeleta of simply connected flag simplicial complexes with 6-large links (i.e. they do not have induced 4- or 5-cycles). Weakly bridged \cite{CO15} graphs generalize bridged graphs and are exactly the weakly modular graphs with convex balls.

Helly graphs are the graphs in which the family of balls satisfy the Helly property, i.e., every family of pairwise intersecting balls has a non-empty intersection. They are the discrete analogues of hyperconvex spaces (complete geodesic metric spaces in which the set of closed balls satisfies the Helly property). Helly graphs have been characterised metrically in the papers \cite{BaPe89, BaPr91} and topologically in \cite{CCHO20}, in a local-to-global way. Pseudo-modular graphs \cite{BM86}, also known as 3-Helly graphs, are the graphs in which any three pairwise intersecting balls have non-empty intersection.

The basis graph $G = G(\mathcal{B})$ of a matroid $\mathcal{B}$ is the graph whose vertices are the bases of $\mathcal{B}$ and edges are the pairs $A, B$ of bases differing by an elementary exchange. It is well known that basis graphs faithfully represent their matroids, thus studying the basis graph amounts to studying the matroid itself. 

Thus, the geodesic halfspace separation problem, which is NP-complete in general graphs, is polynomial for three large classes of graphs. This also implies that the halfspace enumeration problem is output polynomial for these graph classes, and thus answers the question posed in \cite{Che24}.

\section{Preliminaries}
\subsection{General preliminaries and notation}
\textbf{Graphs.} Let $G = (V, E)$ be a simple connected graph.
For any path or cycle $P$ in a graph $G$, the \emph{length} of $P$, denoted by $|P|$, is defined as the number of edges in $P$.
The distance between two vertices $x$ and $y$ in $G$, denoted by $d_G(x, y)$, is the length of a shortest path from $x$ to $y$ in $G$. We use $d(x, y)$ when the graph $G$ is clear from context. The distance between two sets $A, B \subseteq V$ is defined as $d(A, B) \defeq \min_{u \in A, v \in B}{d(u, v)}$. 

For two vertices $x$ and $y$, we use $x \sim y$ to denote that $x$ and $y$ are adjacent, and $x \not \sim y$ to denote that they are not adjacent. We use $I(u, v)$ to denote the set of all vertices on any shortest path from $u$ to $v$. Thus, $I(u, v) = \{x \in V \mid d(u, x) + d(v, x) = d(u, v)\}$.

A subgraph $H$ of $G$ is said to be \emph{isometric} if we have $d_H(u, v) = d_G(u, v)$ for all $u, v \in V(H)$. 

The \emph{neighbourhood} of a vertex $v \in V$ is defined as $N(v) \defeq \{x \in V \mid x \sim v\}$. $N[v] \defeq N(v) \cup \{v\}$ is called the \emph{closed neighbourhood} of $v$. Analogously, the closed neighbourhood of a set $X \subseteq V$ is defined as $N[X] \defeq \bigcup_{x \in X}N[x]$.

For any vertex $v \in V$, the ball of radius $k$ around $v$ is defined as $B_k(v) \defeq \{x \in V \mid d(x, v) \leq k \}$. Similarly, we define the sphere of radius $k$ around $v$ as $S_k(v) \defeq \{x \in V \mid d(x, v) = k \}$.

For any set $X \subseteq V$ and $k \in \mathbb{Z}^+$, we define $B_k(X) \defeq \{v \in V \mid \exists x \in X \text{ such that } d(v, x) \leq k \}$. 
\\ \\
\textbf{Convexity spaces.} A convexity space is a pair $(X, \mathfrak{C})$  where $X$ is a set and $\mathfrak{C}$ is a collection of subsets of $X$ such that $\emptyset, X \in \mathfrak{C}$ and $\mathfrak{C}$ is closed under taking intersection.
 The \emph{convex hull} of a set $S \subseteq X$ is defined as $\mathfrak{c}(S) \defeq \bigcap_{C \in \mathfrak{C} \mid S \subseteq C}S$. It is thus the smallest convex set containing $S$. A convex set $H \in \mathfrak{C}$ is called a \emph{halfspace} if its complement $H^c = X \setminus H$ is also convex.

The \emph{halfspace separation} problem \cite{SHW23} is defined as follows.
\begin{problem}
\problemtitle{HalfspaceSeparation}
\probleminput{Convexity space $(X, \mathfrak{C})$, disjoint sets $A, B \subseteq X$}
\problemquestion{Is there a halfspace $H$ such that $A \subseteq H$ and $B \subseteq H^c$?}
\end{problem}

It is easy to see that a halfspace $H$ separates $A$ and $B$ if and only if it separates $\mathfrak{c}(A)$ and $\mathfrak{c}(B)$.
\\ \\
\textbf{Geodesic convexity.} Let $G$ be a graph. A set $X \subseteq V$ is said to be \emph{geodesically convex} if $I(u, v) \subseteq X$ whenever $u, v \in X$. It is easy to see that this defines a convexity space on the vertices of $G$. Moreover, as $I(u, v)$ can be computed in polynomial time for any $u, v \in V$, the geodesic convex hull of any set $X \subseteq V$ can also be computed in polynomial time.

From now on, we shall use \emph{convex} to denote geodesic convexity unless otherwise specified.

\subsection{Weakly modular graphs, meshed graphs and pseudo-modular graphs}
A graph $G$ is said to be \emph{weakly modular with respect to a vertex $u$} if it satisfies the following conditions.
\begin{itemize}
    \item \emph{Triangle condition} TC($u$): For any two adjacent vertices $v, w$ with $d(u, v) = d(u, w)> 1$, there exists a common neighbour $x$ of $v$ and $w$ such that $d(u, x) = d(u, v)-1$.
    \item \emph{Quadrangle condition} QC($u$): For any three vertices $v, w, z$ with $v \sim z, z \sim w, v \not \sim w$ such that $2 \leq d(u, v) = d(u, w) =d(u, z) -1$, there exists a common neighbour $x$ of $v$ and $w$ such that $d(u, x) = d(u, v) -1$.
\end{itemize}

A graph $G$ is called \emph{weakly modular} \cite{BaCh,Che89} if it is weakly modular with respect to $u$ for all $u \in V$.
\\
Meshed graphs were introduced in the unpublished paper \cite{BMS94} and further studied in \cite{BaCh02, BeCCY23,CCHO20, CCCJ24, Ch_mat}.

A graph $G = (V, E)$ is called \emph{meshed} if for any vertex $u$ its distance function $d$ satisfies the following \emph{weak quadrangle condition}:
\\
(QC$^-$): for any $u, v, w \in V$ with $d(v, w) = 2$, there exists a common neighbour $x$ of $v$ and $w$ such
that $2d(u, x) \leq d(u, v) + d(u, w)$.

(QC$^-$) seems to be a relaxation of (QC), but it implies the triangle condition (TC) (that is not implied by (QC)). Conversely, (TC) and (QC) imply (QC$^-$) and thus weakly modular graphs are meshed.

A connected induced subgraph $H$ of a graph $G$ is said to be \emph{locally convex} \cite{Che89} or \emph{$P_3$-convex} if $I(x, y) \subseteq V(H)$ whenever $x, y \in V(H)$ and $d_G(x, y) = 2$.

We have the following lemma for meshed graphs:
\begin{lem}[\cite{Che24}]
\label{lem:convex-iff-locallyconvex}
A connected induced subgraph $H$ of a meshed graph $G$ is convex if and only if it is locally convex.
\end{lem}

A graph is said to be \emph{pseudo-modular} \cite{BM86} if for any three vertices $u, v, w$ such that $1 \leq d(v, w) \leq 2$ and $d(u, v) = d(u, w) = k \geq 2$, there exists a vertex $x$ such that $x \sim v, w$ and $d(u, x) = k-1$. This condition implies the triangle and quadrangle condition, and thus pseudo-modular graphs are weakly modular. 

\subsection{Helly graphs}
The graphs in which every family of pairwise intersecting balls have a non-empty intersection are known as \emph{Helly} graphs. They are the discrete analogues of hyperconvex (injective) spaces (complete geodesic metric spaces in which the set of closed balls
satisfies the Helly property). In perfect analogy with hyperconvexity, there is a close relationship between Helly graphs and absolute retracts. A graph is an \emph{absolute retract} exactly when it is
a retract of any larger graph into which it embeds isometrically. Absolute retracts and Helly graphs are the same \cite{HR87}. In particular, for any graph $G$ there exists a smallest Helly graph comprising $G$ as an
isometric subgraph \cite{JMP86, Pesch88}. This is the discrete analogue of Isbell's theorem \cite{Isb64} which states that every metric space has an injective hull.
Helly graphs have been characterised metrically in the papers \cite{BaPe89, BaPr91} and topologically in \cite{CCHO20}, in a local-to-global way. 
\begin{lem}[\cite{BM86}]
Helly graphs are pseudo-modular.
\end{lem}
\begin{proof}
Let $u, v, w$ be three arbitrary vertices in a Helly graph $G$ such that $1 \leq d(v, w) \leq 2$ and $d(u, v) = d(u, w) = k \geq 2$. Then the balls $B_1(v), B_1(w)$ and $B_{k-1}(u)$ pairwise intersect. As $G$ is Helly, $B_1(v) \cap B_1(w) \cap B_{k-1}(u)$ is non-empty. Let $x \in B_1(v) \cap B_1(w) \cap B_{k-1}(u)$. Then we must have $d(x, u) = k-1$ and $d(x, v) = d(x, w) = 1$. Thus, $G$ is pseudo-modular.
\end{proof}

In fact, pseudo-modular graphs are exactly the graphs in which  any three pairwise intersecting balls have non-empty intersection \cite{BM86}. They are hence also known as \emph{3-Helly} graphs.

\subsection{Bridged and weakly bridged graphs}
A graph $G$ is said to be \emph{bridged} \cite{FJ87,SoCh} if all isometric cycles of $G$ have length three. A \emph{bridge} $B$ of a cycle $C$ is an isometric (shortest) path between two vertices in $C$ that is shorter than each of the two arcs joining them in $C$. Thus, a graph is bridged if and only if every cycle of length greater than three has a bridge. In particular, all chordal graphs are bridged.

We have the following results about bridged graphs.

\begin{thm}[\cite{Che89}]
   A graph is bridged if and only if it is weakly modular and does not contain $C_4$ or $C_5$ as an induced subgraph. 
\end{thm}

\begin{thm}[\cite{FJ87,SoCh}]
\label{thm:convex-neighbourhood}
   A graph $G$ is bridged if and only if $N[X]$ is g-convex for every g-convex set $X \subseteq V$.
\end{thm}
For any set $X \subseteq V$ and $k \in \mathbb{Z}^+$, we define $B_k(X) \defeq \{v \in V \mid \exists x \in X \text{ such that } d(v, x) \leq k \}$. 
\begin{thm}[\cite{FJ87,SoCh}]
\label{thm:convex-balls}
    A graph $G$ is bridged if and only if $B_k(X)$ is convex for every convex set $X \subseteq V$ for every $k \in \mathbb{Z}^+$.
\end{thm}

Weakly modular graphs in which all balls $B_k(v)$ are convex for all $v \in V$ and $k \in \mathbb{Z}^+$, which are known as \emph{weakly bridged} graphs \cite{osa13, CO15}, generalize bridged graphs. Weakly bridged graphs are exactly the $C_4$-free weakly modular graphs \cite{CO15}.

Let $v$ be a vertex of a graph $G$. We say that $G$ satisfies the property $SD_n(v)$ (simple descent on balls of radii at most $n$ around $v$) \cite{osa13} if for every $i = 1,2,..,n$ the following condition holds: for every $A \subseteq S_{i+1}(v)$ such that $A$ is a clique, there exists a vertex $w$ such that $d(w, v) = i$ and $w$ is adjacent to all vertices in $A$, i.e., $A \cup \{w\}$ is also a clique. We shall say that the property SD$(v)$ is satisfied if SD$_n(v)$ is satisfied for all natural numbers $n \in \mathbb{Z}^+$. Analogously, we shall say that $G$ satisfies the $k$-SD$_n(v)$ (respectively $k$-SD) property if it satisfies $SD_n(v)$ (respectively SD$(v)$) for all cliques $A$ with at most $k$ vertices. Note that 2-SD is the same as the triangle condition TC.

It was shown in \cite{osa13} that weakly bridged graphs satisfy the property SD$_n(v)$ for every natural number $n$ and every vertex $v \in V$.

\subsection{Matroid basis graphs}
According to one of the many equivalent definitions, a \emph{matroid} on a finite set of elements $I$ is a collection $\mathcal{B}$ of subsets of $I$ , called \emph{bases}, which satisfy the following \emph{exchange property}:
\\
(EP) for all $A, B \in \mathcal{B}$ and $i \in A \setminus B$ there exists $j \in B \setminus A$ such that $(A \setminus \{i\}) \cup \{j\}) \in \mathcal{B}$.

We say that the base $(A \setminus \{i\}) \cup \{j\}) $ is obtained from the base $A$ by an \emph{elementary exchange}. It
is well known that all the bases of a matroid have the same cardinality, which is called its \emph{rank}.

The \emph{basis graph} $G = G(\mathcal{B})$ of a matroid $\mathcal{B}$ is the graph whose vertices are the bases of $\mathcal{B}$ and edges are the pairs $A, B$ of bases differing by a single exchange. In other words, $A$ and $B$ are adjacent if and only if $|A \setminus B| = |B \setminus A| = 1$.

In \cite{MBGI}, Maurer characterized the basis graphs of matroids as connected graphs satisfying the following conditions. He also conjectured that the link condition is redundant. This was confirmed by Chalopin, Chepoi and Osajda in \cite{CCO15}  

\begin{enumerate}
    \item \emph{Interval condition} (IC): For any two vertices $u, v \in V$ such that $d(u, v) = 2$, the subgraph induced by $I(u, v)$ is a square, a pyramid or an octahedron (see Figure \ref{fig:matroid}). In particular, $I(u, v)$ always contains a square.
    \item \emph{Positioning condition} (PC): For any square $v_1v_2v_3v_4$ and any vertex $u \in V$, we have $d(u, v_1) + d(u, v_3) = d(u, v_2) + d(u, v_4)$.
    \item \emph{Link condition}: $N(v)$ is the line graph of a bipartite graph for any $v \in V$.
\end{enumerate}
\begin{figure}[!h]
    \centering
    \includegraphics[width=0.6\linewidth]{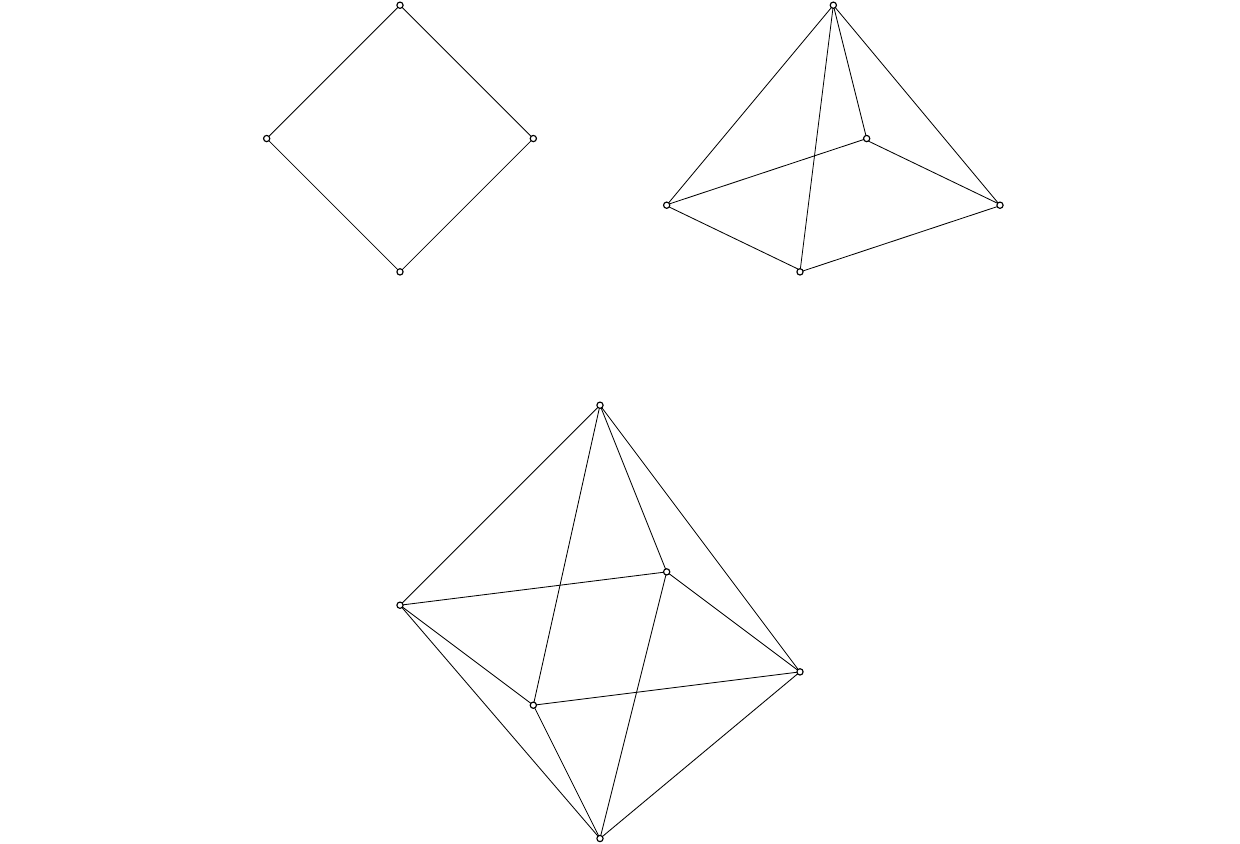}
    \caption{A square, a pyramid and an octahedron}
    \label{fig:matroid}
\end{figure}

\begin{lem}[\cite{Ch_mat}]
Matroid basis graphs are meshed.
\end{lem}
\begin{proof}
Let $u, v, w \in V$ with $d(v, w) = 2$ and let $k := d(u, v) + d(u, w)$. We need to show that there exists $x \sim v, w$ such that $2d(u, x) \leq k$.  If $k = 2$, then $d(u, v) = d(u, w) = 1$ and $u$ itself satisfies our requirements. Suppose $k \geq 3$. By the interval condition IC, $I(v, w)$ contains a square, say $vxwy$. By the positioning condition PC, we have $d(u, x) + d(u, y) = d(u, v) + d(u, w) = k$. Suppose $d(u, x) \leq d(u, y)$. Then we have $2d(u, x) \leq d(u, x) + d(u, y) = k$ and we are done.
\end{proof}

\section{Geodesic halfspace separation}
The \emph{three step method} introduced in \cite{Che24} gives an algorithm for monophonic halfspace separation by first reducing the general halfspace separation problem to the special case of \emph{shadow-closed osculating} pairs of convex sets. They then construct a {\sc 2-SAT} formula to solve monophonic halfspace separation for this special case, the details of which are corrected in the paper \cite{BCET25}. We adapt this method to geodesic halfspace separation in weakly bridged, pseudo-modular and matroid basis graphs. We construct a {\sc 2-SAT} formula for geodesic halfspace separation and prove the following theorem, which is the detailed version of \ref{thm:main}.
\begin{thm}
\label{thm:2-SAT-sep}
Let $(A, B)$ be a pair of shadow-closed osculating convex sets of a graph $G$ and $V^+ \defeq V \setminus (A \cup B)$. If $G$ is weakly bridged, pseudo-modular or the basis graph of a matroid, then the following are equivalent for any partition $(A^+, B^+)$ of $V^+$:
\begin{enumerate}
\item $H = A \cup A^+$, $H^c = B \cup B^+$ are complementary halfspaces.
\item $A^+ = \{x \in V^+ \mid \alpha(a_x) = 1\}$ and $B^+ = \{x \in V^+ \mid \alpha(a_x) = 0\}$ for a solution $\alpha$ of the {\sc 2-SAT} formula $\Phi$ defined in \ref{sec:2-SAT-red}.
\end{enumerate}
\end{thm}

\subsection{The three step method}
Let $G = (V, E)$ be a simple connected graph and $A, B \subseteq V$. We shall first find the vertices that can be easily seen to lie on the same side as $A$ (or $B$) in any pair of complementary halfspaces.

Let $P = u_1,..,u_k$ be any shortest path from $A$ to $B$. If $H, H^c$ are complementary halfspaces such that $A \subseteq H$ and $B \subseteq H^c$, then we must have an edge $u_iu_{i+1}$ in $P$ with $u_i \in H$ and $u_{i+1} \in H^c$. Thus, for any halfspace $H$,  $A \subseteq H$ and $B \subseteq H^c$ if and only if $A \cup \{u_i\} \subseteq H$ and $B \cup \{u_{i+1}\} \subseteq H^c$ for some $i \in \{1,..,k-1\}$. A pair of sets $(A, B)$ is said to be osculating if $A$ and $B$ are disjoint and adjacent. The first step of our algorithm is to use the above argument to produce instances $(A_i, B_i)$ of osculating pairs.

The \emph{shadow of $A$ with respect to $B$} \cite{Ch_sep} is the set $$A/B \defeq \{x \in V \mid \mathfrak{c}(B \cup \{x\}) \cap A \neq \emptyset \}$$. 

Observe that $A \subseteq A/B$ and $B \subseteq B/A$. Thus, if $H$ is a halfspace that separates $A/B$ and $B/A$, we have $A \subseteq H$ and $B \subseteq H^c$. Conversely, let $H$ be a halfspace with $A \subseteq H$ and $B \subseteq H^c$. For any vertex $v$, we have $v \in H \implies \mathfrak{c}(A \cup \{v\}) \subseteq H$ and $v \in H^c \implies \mathfrak{c}(B \cup \{v\}) \subseteq H^c$. We conclude that a halfspace $H$ separates $A$ and $B$ if and only if it separates $A/B$ and $B/A$. Thus, we have the following lemma.
\begin{lem}[\cite{Che24}]
\label{lem:shadow-closure}
Let $A, B \subseteq V$ and $H$ be a halfspace in $G$. $A \subseteq H$ and $B \subseteq H^c$ if and only if $\mathfrak{c}(A/B) \subseteq H$ and $\mathfrak{c}(B/A)\subseteq H^c$.
\end{lem}

Thus, $H$ separates $A$ and $B$ if and only if it separates $\mathfrak{c}(A/B)$ and $\mathfrak{c}(B/A)$. Now, we iterate this argument. Let $A^0 = A$ and $B^0 = B$. For $i \geq 1$, we define $A^i \defeq \mathfrak{c}(A^{i-1}/B^{i-1})$ and $B^i \defeq \mathfrak{c}(B^{i-1}/A^{i-1})$. Since $A^i \subseteq \mathfrak{c}(A^i/B^i) = A^{i+1}$ and $B^i \subseteq \mathfrak{c}(B^i/A^i) = B^{i+1}$ for all $i$, the sequences $(A^i)_{i \geq 0}$ and $(B^i)_{i \geq 0}$ must converge to some sets $A^*$ and $B^*$ with $A^* = A^*/B^*$ and $B^* = B^*/A^*$. We conclude that a halfspace $H$ separates $A$ and $B$ if and only if it separates $A^*$ and $B^*$. The pair $(A^*, B^*)$ is called the \emph{shadow-closure} of $(A, B)$, and we say that a pair of sets $(A, B)$ is \emph{shadow-closed} whenever $A = A/B$ and $B = B/A$.
\begin{algorithm}
\caption{ \sc ShadowClosure}
\begin{algorithmic}
\State \textbf{Input: } Graph $G$, disjoint sets $A, B \subseteq V$
\State \textbf{Output: } The shadow-closure of $(A, B)$
\While{ $A \neq A/B$ or $B \neq B/A$}
    \State $A \gets A/B, B \gets B/A$
\EndWhile
\State \Return $(A, B)$
\end{algorithmic}
\end{algorithm}

\begin{figure}[ht]
    \centering
    \includegraphics[width=0.3\linewidth]{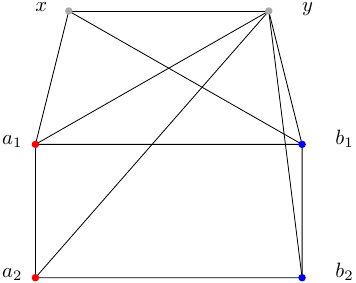}
    \caption{A shadow-closed osculating pair. The grey vertices $x$ and $y$ are the residue $V^+$}
    \label{fig:sc-osc}
\end{figure}
Combining the above two ideas, the {\sc HalfspaceSeparation} problem for a graph $G$ and sets $A, B \subseteq V$ can be reduced to the case where $A$ and $B$ are a shadow-closed osculating pair of convex sets. In the next subsection, we construct a {\sc 2-SAT} formula $\Phi$ that solves geodesic halfspace separation problem for shadow-closed osculating pairs of convex sets when $G$ is weakly bridged, pseudo-modular or the basis graph of a matroid. 

Our algorithm is as follows.
\begin{algorithm}
\caption{\sc HalfspaceSep}
\begin{algorithmic}
\State \textbf{Input:} Weakly bridged, pseudo-modular or matroid basis graph $G$, disjoint sets $A, B \subseteq V$
\State \textbf{Question:} Is there a halfspace $H$ such that $A \subseteq H$ and $B \subseteq H^c$?

\State Compute a shortest $A-B$ path $P = u_1..u_k$.
\For{$i = 1,..,k-1$}
    \State $A_i, B_i = ${\sc ShadowClosure}$(A \cup \{u_i\}, B \cup \{u_{i+1}\})$
    \State Compute the {\sc 2-SAT} formula $\Phi$ defined in \ref{sec:2-SAT-red} corresponding to $(A_i, B_i)$.
    \If{$\Phi$ has a satisfying assignment}
        \State \Return YES
    \EndIf
\EndFor
\\
\Return NO
\end{algorithmic}
\end{algorithm}
As {\sc 2-SAT} can be solved in linear time, this implies that {\sc HalfspaceSeparation} is polynomial for these classes of graphs.

\subsection{Reduction to {\sc 2-SAT}}
\label{sec:2-SAT-red}
Let $A, B \subseteq V$ be a shadow-closed osculating pair of convex sets.
 We define the \emph{residue} as $V^+ \defeq V \setminus  (A \cup B)$. Note that if $V^+ = \emptyset$, then we can answer yes to {\sc HalfspaceSeparation}$(A, B)$. Thus, we may assume that $V^+ \neq \emptyset$ and our goal is to distribute the vertices in the residue into complementary halfspaces. \cite{BCET25} gives a polynomial algorithm for monophonic {\sc HalfspaceSeparation} using a similar approach, and reduce the problem to {\sc 2-SAT} by associating a boolean variable with each vertex in the residue. In this section, we shall construct a 2-SAT formula to solve the geodesic halfspace separation problem in weakly bridged graphs, pseudo-modular graphs, and matroid basis graphs.

 First, we shall prove some properties of shadow-closed osculating pairs.

\begin{prop}
\label{prop:equidistant}
   Let $x \in V^+$ and $ab \in E$ with $a \in  A$ and $b \in B$. Then $x$ is equidistant from $a$ and $b$.  
\end{prop}
\begin{proof}
Without loss of generality let $k = d(x, a) \geq d(x, b)$. Then $d(x, b)$ must be $k-1$ or $k$. Suppose it is $k-1$. Then $b \in I(a, x)$. Thus $x \in A/B$ which contradicts the fact that $A$ is shadow-closed. Thus we must have $d(x, b) = k = d(x, a)$.
\end{proof}

For $x \in V^+$ and $ab \in E$ with $a \in A$ and $b \in B$, we define the set $$S_x^{ab} \defeq \set{x_0 \in V^+}{x_0 \sim a, x_0 \sim b,  x_0 \in I(x, a) \cap I(x, b)  }.$$ We define $ S_x \defeq \bigcup_{ab \in E, a \in A, b \in B}{S_x^{ab}}$.

\begin{prop}
\label{prop:Sx-non-empty}
    For any $x \in V^+$ and any edge $ab \in E$ with $a \in A$, $b \in B$, $S_x^{ab}$ is non-empty if $G$ satisfies the triangle condition TC.
\end{prop}

\begin{proof}
    Let $x \in V^+$ and $ab$ be any edge with $a \in A$ and $b \in B$. By Proposition \ref{prop:equidistant}, $d(x, a) = d(x, b)$. Let this distance be $k$. If $k = 1$, then $x \in S_x^{ab}$ and we are done. Suppose $k > 1$. By the triangle condition, there exists $w \in V$ such that $w \sim a, w \sim b$ and $d(x, w) = k-1$. If $w \in V^+$, we are done. Suppose not. Without loss of generality, assume $w \in A$. Then by Proposition \ref{prop:equidistant}, $d(x, b) = d(x, w) = k-1$, which is a contradiction. Thus, $w \in V^+$, which implies $w \in S_x^{ab}$.
\end{proof}
In particular, since $A$ and $B$ osculate, we have $S_x \neq \emptyset$ for any $x \in V^+$ in any graph satisfying TC.

To distribute the residue $V^+$ into complementary halfspaces $H$ and $H^c$, we define a 2-SAT formula $\Phi$, which shall be a necessary condition for {\sc HalfspaceSeparation}.
For each vertex $x \in V^+$, we define a binary variable $a_x$.

For any $x \in V^+$ and $x_0 \in S_x$, we have $x_0 \in \mathfrak{c}(x \cup A) \cap \mathfrak{c}(x \cup B)$. Thus, if $H, H^c$ are complementary halfspaces separating $A$ and $B$, we must have $x \in H$ if and only if $x_0 \in H$. Consequently, if $x, y \in V^+$ and $S_x \cap S_y \neq \emptyset$, we want $x$ and $y$ to be in the same halfspace. Thus, any $z \in I(x, y)$ must also be in the same halfspace as $x$ and $y$.

We thus want the set of variables $\{a_x\}_{x \in V^+} $ to satisfy the following constraints, called the \emph{equality constraints}:  for all $x, y, z \in V^+$ such that $S_x \cap S_y \neq \emptyset$ and $z \in I(x, y)$, we have $a_z  = a_{x_0}$ for any $x_0 \in S_x \cap S_y$.

These equality constraints can be encoded by the following {\sc 2-SAT} formula:
\[
\Phi_1  = 
\bigwedge_{\substack{x, y, z \in V^+\\ x_0 \in S_x \cap S_y  \\ z \in I(x, y)}}(a_{x_0} \lor \overline{a_{z}}) \land (\overline{a_{x_0}} \lor a_{z})
\]

Secondly, we want the vertices $x \in V^+$ and their binary variables $a_x$ to satisfy the following conditions, called the \emph{implication constraints}:
\begin{enumerate}
    \item If $x \in V^+,\ z \in A$ and $y \in I(x, z)$ then $x \in H \implies y \in H$,
    \item If $x \in V^+,\ z \in B$ and $y \in I(x, z)$, then $x \in H^c \implies y \in H^c$,
\end{enumerate}
where $H, H^c$ are complementary halfspaces such that $A \subseteq H$ and $B \subseteq H^c$.
For vertices $x, y \in V^+$, we shall say that $x$ $A$-implies $y$, written $x \to_A y$, if there exists a vertex $z \in A$ such that $y \in I(x, z)$. Analogously, we say that \emph{x $B$-implies y}, denoted by $x \to_B y$, if there exists $z \in B$ such that $y \in I(x, z)$.
These implication constraints can be encoded by the following {\sc 2-SAT} formula:
\[
\Phi_2 = \bigwedge_{\substack{x, y \in V^+\\ x \rightarrow_A y}} (\overline{a_x} \lor a_y)~~\land \bigwedge_{\substack{x, y \in V^+\\ x {\rightarrow_B} y}} (a_x \lor \overline{a_y}).
\]
\begin{rem}
    For any $x \in V^+$ and $x_0 \in S_x$, we have $x \rightarrow_A x_0$ and $x \rightarrow_B x_0$. Thus, the implication constraints ensure that $x$ and $x_0$ will be in the same halfspace.
\end{rem}
Finally, we want the vertices to satisfy the following constraints:
\begin{enumerate}
    \item If $x, y \in V^+$ and $I(x, y) \cap B \neq \emptyset$, then at least one of $x, y$ must be in $H^c$,
    \item If $x, y \in V^+$ and $I(x, y) \cap A \neq \emptyset$, then at least one of $x, y$ must be in $H$.
\end{enumerate}
These can be encoded by the following {\sc{2-SAT}} formula:
\[
\Phi_3 = \bigwedge_{\substack{x, y \in V^+\\ I(x,y)\cap B \neq \emptyset}} (\overline{a_x} \lor \overline{a_y})~~\land \bigwedge_{\substack{x, y \in V^+\\ I(x, y) \cap A \neq \emptyset}}(a_x \lor a_y).
\]
Let $\Phi \defeq \Phi_1 \land \Phi_2 \land \Phi_3$. For any complementary halfspaces $(H, H^c)$, consider the assignment $\alpha$ that assigns $a_x = 1$ for all $x \in V^+ \cap H$ and $a_x = 0$ for all $x \in V^+ \cap H^c$. 
Note that if $(H, H^c)$ are complementary halfspaces separating $A$ and $B$, then by construction $\alpha$ must be a satisfying assignment for~$\Phi$. In other words, the satisfiability of $\Phi$ is a necessary condition for halfspace separation for any graph. We shall show that this is  also a sufficient condition for halfspace separation if $G$ is weakly bridged, pseudo-modular or the basis graph of a matroid. 

As {\sc 2-SAT} can be solved in linear time, this proves that {\sc HalfspaceSeparation} is polynomial for these classes of graphs.

From now on, we assume that $G$ satisfies the triangle condition.

Let $\alpha$ be any satisfying assignment of $\Phi$. Let $A^+ \defeq \{x \in V^+ \mid \alpha(a_x) = 1\}$ and $B^+ \defeq \{x \in V^+ \mid \alpha(a_x) = 0\}$. Consider the sets $H \defeq A \cup A^+$ and $H^c = B \cup B^+$.
\begin{prop}
\label{prop:H-is-connected}
    The graphs induced by $H$ and $H^c$ are connected.
\end{prop}
\begin{proof}
    $A$ and $B$ are connected as they are convex. Consider any vertex $x \in V^+$ and edge $ab$ with $a \in A$ and $b \in B$. By Proposition \ref{prop:Sx-non-empty}, $S_x^{ab}$ is non-empty. Let $x_0 \in S_x^{ab}$. For any $x' \in I(x, x_0)$, we have $x' \in I(x, a) \cap I(x, b)$. Since $A$ and $B$ are shadow-closed, we must have $x' \in V^+$ and by the implication constraints, we have $x' \in A^+$ if and only if $x \in A^+$. Thus, we have $I(x, x_0)\subseteq A^+$ if $x \in A^+$ and $I(x, x_0) \subseteq B^+$ if $x_0 \in B^+$. Thus, $H$ and $H^c$ are connected. 
\end{proof}

Thus, by Lemma \ref{lem:convex-iff-locallyconvex}, if $G$ is meshed then $H$ and $H^c$ are complementary halfspaces if and only if they are locally convex. In particular, this holds for weakly bridged graphs, pseudo-modular graphs and matroid basis graphs.

In the rest of this section, we shall show that the sets $H, H^c$  given by any solution of the {\sc 2-SAT} $\Phi$ are locally convex in these classes of graphs. As {\sc 2-SAT} can be solved in linear time, this shows that {\sc HalfspaceSeparation} is polynomial for them.

We require the following useful lemma.
\begin{lem}
\label{lem:distances-equal-in-countereg}
    If $H$ is not locally convex, then there exist $x, y \in A^+$ and $z \in B^+$ such that $d(x, y) = 2$ and $z \in I(x, y)$. Moreover, for any edge $ab$ with $a \in A$ and $b \in B$, we must have $d(x, a) = d(y, a) = d(z, a)= d(x, b) = d(y, b) = d(z, b) \geq 2$.
\end{lem}
 \begin{proof}
     If $H$ is not locally convex, there exist $x, y \in H$ and $z \in H^c$  such that $d(x, y) = 2$ and $z \in I(x, y)$. As $A$ is convex, $x$ and $y$ cannot both be in $A$. If $x \in A^+, y \in A$ and $z \in B^+$, the implication constraints are violated and $\alpha$ is not a satisfying assignment. If $x \in A^+, y \in A$ and $z \in B$, we get $x \sim y$ by Proposition \ref{prop:equidistant}, which contradicts our assumption that $d(x, y) = 2$. If $x, y \in A^+$ and $z \in B$, one of the clauses in $\Phi_3$ is not satisfied. Thus, we must have $x, y \in A^+$ and $z \in B^+$. 

     Let $ab$ be an edge with $a \in A$ and $b \in B$. By Proposition \ref{prop:equidistant}, we have $d(x, a) = d(x, b), d(y, a) = d(y, b)$ and $d(z, a) = d(z, b)$.  Let $d(x, a) = k$. By the triangle inequality, $k-1 \leq d(z, a) \leq k+1$. If $d(z, a) = k-1$, then $z \in I(x, a)$, i.e., $x \rightarrow_A z$ and by the implication constraints we get $z \in A^+$, which is a contradiction. Similarly, if $d(z, a) = d(z, b) = k+1$, we get $z \rightarrow_B x$, whence $x \in B^+$, which is a contradiction. Thus, $d(z, a) = d(z, b) = d(x, a) = d(x, b) = k$, and by symmetry we also have $d(y, a) = d(y, b) = k$. If $k = 1$, we get $b \in I(x, y)$, which violates the constraints of $\Phi_3$. Thus, $k \geq 2$.
 \end{proof}

\subsection{Halfspace separation in bridged graphs}
The case of bridged graphs is settled by the following lemma.
\begin{lem}
\label{lem:bridged}
If $G$ is a bridged graph, then any solution to the {\sc 2-SAT} $\Phi$ corresponds to a solution to {\sc HalfspaceSeparation}.
\end{lem}
\begin{proof}
Let $H$ and $H^c$ be the partition of $V$ given by a solution of $\Phi$. We shall prove by contradiction that $H$ (and $H^c$) are convex. Suppose $H$ is not convex. By Proposition \ref{prop:H-is-connected}, it is connected. Hence, by Lemma \ref{lem:convex-iff-locallyconvex}, $H$ is not locally convex. By Lemma \ref{lem:distances-equal-in-countereg}, there exist $x, y \in A^+$ and $z \in B^+$ such that $z \in I(x, y)$, $d(x, y) = 2$, and for any edge $ab$ with $a \in A$ and $b \in B$, we have $d(x, a) = d(x, b) = d(y, a) = d(y, b) = d(z, a) = d(z, b)$. By the equality constraints, we know that $S_x \cap S_y = \emptyset$. Fix an edge $ab$ with $a \in A$ and $b \in B$ and let $x_0 \in S_x^{ab}$ and $y_0 \in S_y^{ab}$. We know that $x_0, y_0 \in A^+$. If $x_0 \not \sim y_0$, we would have $b \in I(x_0, y_0)$, which violates $\Phi_3$. Thus, we have $x_0 \sim y_0$. Let $d(x, x_0) = d(y, y_0) = k$. We must have $k \leq d(z, x_0), d(z, y_0) \leq k+1$. But if $d(z, x_0) = k$, we would have $x_0 \in S_z$ which is not possible. Hence $d(z, x_0) = d(z, y_0) = k+1$.

 Let $P_x$ and $P_y$ denote shortest paths from $x$ to $x_0$ and $y$ to $y_0$, respectively. Consider the ball $B_k(\{ x_0,y_0\})$ around the convex set $\{ x_0,y_0\}$. Since $d(x_0,x)=d(y,y_0)=k$, the vertices $x$ and $y$ belong to $B_k(\{ x_0,y_0\})$. Since $G$ is bridged,  $B_k(\{ x_0,y_0\})$ is convex by Theorem~\ref{thm:convex-balls}. Since $z\in I(x,y)$, we conclude that $z\in B_k(\{ x_0,y_0\})$, which contradicts the fact that $d(z, x_0) = d(z, y_0) = k+1$.
\end{proof}
\subsection{Halfspace separation in weakly bridged graphs}
Recall that a graph $G$ is said to satisfy the property SD($v$) (simple descent on balls around v) if for every $i \in \mathbb{N}$, the following condition holds: for every $A \subseteq S_{i+1}(v)$
such that $A$ is a clique, there exists a vertex $w$ such that $d(w, v) = i$ and $w$ is adjacent to all vertices in $A$, i.e.,
$A \cup \{w\}$ is also a clique. $G$ is said to satisfy the $k$-SD property if it satisfies SD for all cliques $A$ with at most $k$ vertices.
\begin{lem}
    If $G$ satisfies 3-SD, then any solution of $\Phi$ partitions $V$ into two locally convex sets.
\end{lem}
\begin{proof}
  Let $H$ and $H^c$ be the partition of $V$ given by a solution of $\Phi$. We shall prove by contradiction that $H$ (and $H^c$) are locally convex. Suppose $H$ is not locally convex. By Lemma \ref{lem:distances-equal-in-countereg}, there exist $x, y \in A^+$ and $z \in B^+$ such that $z \in I(x, y)$, $d(x, y) = 2$, and for any edge $ab$ with $a \in A$ and $b \in B$, we have $d(x, a) = d(x, b) = d(y, a) = d(y, b) = d(z, a) = d(z, b) \geq 2$. By the equality constraints, we know that $S_x \cap S_y = \emptyset$. Fix an edge $ab$ with $a \in A$ and $b 
    \in B$ and let $d(x, a) = d(x, b) = d(y, a) = d(y, b) = d(z, a) = d(z, b) = k$.

    By the triangle condition TC (which is implied by 3-SD), $x$ and $z$ have a common neighbour $u$ such that $d(u, a) = k-1$. If $u \in V^+$, we have $S_u^{ab} \subseteq S_x \cap S_z$, which is a contradiction as $\alpha(a_x) \neq \alpha(a_z)$. If $u \in B$, we have $ \mathfrak{c}(A \cup \{x\}) \cap B \neq \emptyset$ and thus $x \in B/A$, which contradicts the fact that $B$ is shadow-closed. Thus, we must have $u \in A$. If $d(u, b) = k-1$, we get $x, z \in A/B$ which is a contradiction. Hence we must have $d(u, b) = k$. If $u \not \sim y$, we have $z \in I(y, u)$, i.e., $y \rightarrow_A z$ and the implication constraints would be violated. Hence, we have $u \sim y$.
    
    As $G$ satisfies the 3-SD condition, $x, z$ and $u$ have a common neighbour $w$ such that $d(w, b) = k-1$. If $w \in V^+$, we would have $S_w^{ab} \subseteq S_x^{ab} \cap S_z^{ab}$, which is a contradiction. If $w \in A$, we would have $x, z \in A/B$, which contradicts the fact that $A$ is shadow-closed. If $w \in B$, then $uw$ is an edge from $A$ to $B$ and by Proposition \ref{prop:equidistant} we would have $y \sim w$. But then $w \in I(x, y) \cap B$ and the constraint $\Phi_3$ is violated.
  
\end{proof}

Thus, if $G$ also has the property given by Theorem \ref{lem:convex-iff-locallyconvex}, then the solutions of the 2-SAT are exactly the geodesic complementary halfspaces separating $A$ and $B$. This is the case for weakly bridged graphs since they are weakly modular and hence meshed. Thus, we have the following corollary which settles halfspace separation for weakly bridged graphs.
\begin{lem}
\label{lem:weakly-bridged}
If $G$ is a weakly bridged graph, then any solution to the {\sc 2-SAT} $\Phi$ corresponds to complementary halfspaces separating $A$ and $B$.
\end{lem}
\subsection{Halfspace separation in pseudo-modular graphs}
The case of pseudo-modular graphs is settled by the following lemma.
\begin{lem}
\label{lem:pseudo-modular}
If $G$ is a pseudo-modular graph, then any solution to the {\sc 2-SAT} $\Phi$ corresponds to a solution to {\sc HalfspaceSeparation}.
\end{lem}
\begin{proof}
    Let $H$ and $H^c$ be the partition of $V$ given by a solution of $\Phi$. We shall prove by contradiction that $H$ (and $H^c$) are convex. Suppose $H$ is not convex. By Proposition \ref{prop:H-is-connected}, it is connected. Hence, by Lemma \ref{lem:convex-iff-locallyconvex}, $H$ is not locally convex. By Lemma \ref{lem:distances-equal-in-countereg}, there exist $x, y \in A^+$ and $z \in B^+$ such that $z \in I(x, y)$, $d(x, y) = 2$, and for any edge $ab$ with $a \in A$ and $b \in B$, we have $d(x, a) = d(x, b) = d(y, a) = d(y, b) = d(z, a) = d(z, b) \geq 2$. By the equality constraints, we know that $S_x \cap S_y = \emptyset$. Fix an edge $ab$ with $a \in A$ and $b 
    \in B$ and let $d(x, a) = d(x, b) = d(y, a) = d(y, b) = d(z, a) = d(z, b) = k$.

    Since $d(x, y) = 2$ and $d(x, b) = d(y, b) = k \geq 2$, by the definition of pseudo-modular graphs, $x$ and $y$ have a common neighbour $u$ such that $d(u, b) = k-1$. If $u \in V^+$, then we have $d(u, a) = d(u, b)= k-1$ and $S_u^{ab} \subseteq S_x \cap S_y$, which is a contradiction. If $u \in A$, then $u \in \mathfrak{c}(x \cup B) \cap A$, and thus $x \in A/B$ which contradicts the fact that $A$ is shadow-closed. Thus, we must have $u \in B$. Thus $I(x, y) \cap B \neq \emptyset$, which violates $\Phi_3$.
  
\end{proof}
\subsection{Halfspace separation in matroid basis graphs}
The case of matroid basis graphs is settled by the following lemma.
\begin{lem}
\label{lem:mbg}
If $G$ is a matroid basis graph, then any solution to the {\sc 2-SAT} $\Phi$ corresponds to a solution to {\sc HalfspaceSeparation}.
\end{lem}
\begin{proof}
    Let $H$ and $H^c$ be the partition of $V$ given by a solution of $\Phi$. We shall prove by contradiction that $H$ (and $H^c$) are convex. Suppose $H$ is not convex. By Proposition \ref{prop:H-is-connected}, it is connected. Hence, by Lemma \ref{lem:convex-iff-locallyconvex}, $H$ is not locally convex. By Lemma \ref{lem:distances-equal-in-countereg}, there exist $x, y \in A^+$ and $z \in B^+$ such that $z \in I(x, y)$, $d(x, y) = 2$, and for any edge $ab$ with $a \in A$ and $b \in B$, we have $d(x, a) = d(x, b) = d(y, a) = d(y, b) = d(z, a) = d(z, b) \geq 2$. Fix an edge $ab$ with $a \in A$ and $b \in B$ and let $d(x, a) = d(x, b) = d(y, a) = d(y, b) = d(z, a) = d(z, b) = k \geq 2$.

    By the triangle condition TC, $x$ and $z$ have a common neighbour $u$ such that $d(u, a) = k-1$. If $u \in V^+$, we have $S_u^{ab} \subseteq S_x \cap S_z$, which is a contradiction as $\alpha(a_x) \neq \alpha(a_z)$. If $u \in B$, we have $ \mathfrak{c}(A \cup \{x\}) \cap B \neq \emptyset$ and thus $x \in B/A$, which contradicts the fact that $B$ is shadow-closed. Thus, we must have $u \in A$. If $d(u, b) = k-1$, we get $x, z \in A/B$ which is a contradiction. Hence we must have $d(u, b) = k$. If $u \not \sim y$, we have $z \in I(y, u)$, i.e., $y \rightarrow_A z$ and the implication constraints would be violated. Hence, we have $u \sim y$. Thus, $z, u \in I(x, y)$. By the interval condition IC, we know that $I(x, y)$ contains a square. We thus have the following two cases.

    \textbf{Case 1}: $u$ is in a square in $I(x, y)$.

    Let $z' \in I(x, y)$ such that $xz'yu$ is a square. By the positioning condition PC, we have $d(z', a) = k+1$ and $d(z', b) = k$. Thus, $z' \not \in V^+$ by Proposition \ref{prop:equidistant}. If $z' \in B$, $\Phi_3$ is violated. If $z' \in A$, then $x, y \in I(u, z')$ which contradicts the convexity of $A$.

    \textbf{Case 2}: $u$ is not in a square.

    In this case, $I(x, y)$ is neither a square nor an octahedron. Thus, by IC, $I(x, y)$ must be a pyramid. Let $z' \in I(x, y)$ such that $xzyz'$ is a square. By PC, $d(z', a) = d(z', b) = k$. Since $I(x, y)$ is a pyramid, we have $u \sim z'$ and $\{x, y, z, z', u\} = I(x, y)$. Hence $x, y, u \in I(z, z')$.  

    If $z$ and $z'$ have four neighbours in common, then $I(z, z')$ must be an octahedron. But this implies that $I(x, y)$ is also an octahedron which is a contradiction. Thus, we must have $\{x, y, u, z, z'\} = I(z, z')$.
    
    By TC, $z$ and $y$ have a common neighbour $v$ such that $d(v, b) = k-1$. If $v \in V^+$, we get $S_z \cap S_y \neq \emptyset$ and if $v \in A$, we get $z, y \in A/B$, neither of which is possible. Thus, we must have $v \in B$. If $d(v, a) = k-1$, we get $z, y \in B/A$ which is a contradiction. Hence $d(v, a) = k$. 
    
   If $x \sim v$, we would have $v \in I(x, y)$ which is a contradiction. If $u \sim v$, then by Proposition \ref{prop:equidistant} we would have $x \sim v$, which is not possible. If $z' \sim v$, we would have $v \in I(z, z')$ which is a contradiction. Thus, $d(x, v) = d(z', v) = d(u, v) = 2$.
    
    By TC, $x, z'$ and $v$ must have a common neighbour, say $w$. Note that $w \not \in \{u, y, z\}$. If $w \sim y$, we would have $w \in I(x, y)$ which is not possible. Hence $z'yvw$ is a square, and by PC we have $d(w, a) = k$ and $d(w, b) = k-1$. Thus, $w \not \in V^+$ and $w \in I(x, b)$, which imply that $w \in B$. 

    If $z \sim w$, we get $w \in I(z, z')$ which is a contradiction. Hence $z \not \sim w$. But this implies $z \rightarrow_B x$ and the implication constraints are violated.
\end{proof}
\subsection{Proof of Theorem \ref{thm:2-SAT-sep}}
To conclude the proof of Theorem \ref{thm:2-SAT-sep}, we show that the set of halfspaces separating the sets $A$ and $B$ is in bijection with the set of satisfying assignments of the {\sc 2-SAT} $\Phi$ whenever $G$ is weakly bridged, pseudo-modular or the basis graph of a matroid.
\begin{proof}
Let $V^+ = V \setminus (A \cup B)$. If $(H, H^c)$ be a pair of complementary halfspaces separating $A$ and $B$, then by construction the assignment $\alpha$ that assigns $a_x = 1$ for all $x \in V^+ \cap H$ and $a_x = 0$ for all $x \in V^+ \cap H^c$ is a solution of $\Phi$. Conversely, let $\alpha$ be any satisfying assignment to $\Phi$ and let $A^+ \defeq v\in V^+ \mid \alpha(v) = 1\}$ and $B^+ \defeq V^+ \setminus A^+$. Consider the sets $H \defeq A \cup A^+$ and $H^c = B \cup B^+$. By Lemmas \ref{lem:weakly-bridged}, \ref{lem:pseudo-modular} and \ref{lem:mbg}, $H$ and $H^c$ are complementary halfspaces.
\end{proof}
Thus, {\sc{HalfspaceSeparation}} can be solved in polynomial time for weakly bridged graphs, pseudo-modular graphs and matroid basis graphs. Consequently, halfspaces can be enumerated in output polynomial time for these graph classes.
\begin{cor}
Halfspaces can be enumerated in output polynomial time for the geodesic convexity of weakly bridged, pseudo-modular and matroid basis graphs.
\end{cor}

\begin{proof}
We use the classical enumeration method called backtrack or \emph{flashlight search} \cite{RT75} which we briefly explain here. Let $\mathcal{P}$ be the set of all solutions we want to enumerate. The algorithm is a DFS traversal of a partial solutions tree. Given two disjoint sets $A$ and $B$, the {\sc Extension} problem asks whether there exists a solution $S \in \mathcal{P}$ with $A \subseteq S$ and $B \cap S = \emptyset$. The flashlight search algorithm runs in output polynomial time whenever {\sc Extension} can be solved in polynomial time. When $\mathcal{P}$ is the set of halfspaces of a graph, $H^c \in \mathcal{P}$ whenever $H \in \mathcal{P}$ and thus {\sc Extension} is the same as {\sc HalfspaceSeparation}.
\end{proof}

\bibliographystyle{abbrv}
\bibliography{refs}
\end{document}